\documentclass[12pt]{article}

\usepackage{amsmath}
\usepackage{amssymb}
\usepackage{amsfonts}
\usepackage{latexsym}

\catcode `\@=11 \@addtoreset{equation}{section}

\catcode `\@=12

\newcommand{\be}{\begin{equation}}
\newcommand{\en}{\end{equation}}
\newcommand{\bea}{\begin{eqnarray}}
\newcommand{\ena}{\end{eqnarray}}
\newcommand{\beano}{\begin{eqnarray*}}
\newcommand{\enano}{\end{eqnarray*}}
\newcommand{\bee}{\begin{enumerate}}
\newcommand{\ene}{\end{enumerate}}

\newcommand{\A}{{\mathfrak A}}

\newcommand{\Ao}{{\mathfrak A}_0}

\newcommand{\mc}{\mathcal}

\newcommand{\D}{{\mc D}}
\newcommand{\LL}{\mc L}

\newcommand{\Sc}{{\cal S}}
\newcommand{\E}{{\cal E}}
\newcommand{\F}{{\cal F}}

\newcommand{\Lc}{{\cal L}}
\newcommand{\C}{{\cal C}}
\newcommand{\1}{1 \!\! 1}

\newcommand{\LD}{{\LL}^\dagger (\D)}

\newcommand{\restr}{\upharpoonright}
\newcommand{\Hil}{\mc H}

\newtheorem{thm}{Theorem}
\newtheorem{cor}[thm]{Corollary}
\newtheorem{lemma}[thm]{Lemma}
\newtheorem{prop}[thm]{Proposition}

\newenvironment{proof}{\noindent {\bf Proof --}}{\hfill$\square$ \vspace{3mm}\endtrivlist}

\catcode `\@=11 \@addtoreset{equation}{section}
\catcode `\@=12

\begin{document}

\thispagestyle{empty}

\vspace*{1cm}

\begin{center}
{\Large \bf $O^\star$-algebras and quantum dynamics:
some existence results}   \vspace{2cm}\\

{\large F. Bagarello}
%\footnote[1]{ Dipartimento di Matematica ed Applicazioni,
%Facolt\`a di Ingegneria, Universit\`a di Palermo, I-90128  Palermo, Italy}
\vspace{3mm}\\
  Dipartimento di Metodi e Modelli Matematici,
Facolt\`a di Ingegneria, Universit\`a di Palermo, \\Viale delle Scienze, I-90128  Palermo, Italy\\
e-mail: bagarell@unipa.it\\home page: www.unipa.it/$^\sim$bagarell
\vspace{4mm}\\

\end{center}

\vspace*{2cm}

\begin{abstract}
\noindent We discuss the possibility of defining an algebraic
dynamics within the settings of $O^\star$-algebras. Compared with our previous results on this subject,
the main improvement here is that we are not assuming the existence of some hamiltonian for the {\em full} physical system. We will show that, under suitable conditions, the dynamics can still be defined via some limiting procedure starting from a given {\em regularized sequence}.
\end{abstract}

\newpage
\section{Introduction and mathematical framework}

In a series of previous papers, see \cite{bag1}-\cite{{bagtralp}}
and \cite{bagrev} for an up-to-date review, we have discussed
 the possibility of getting a rigorous
definition of the algebraic dynamics $\alpha^t$ of some physical system by
making use of the so-called quasi *-algebras, as well as some purely
mathematical features of these algebras. In particular we have shown
that, if the hamiltonian $H$ of the system exists and is a
self-adjoint operator, then we can use $H$ itself to build up a quasi
*-algebra of operators and a {\em physical topology} in terms of
which the time evolution of each observable of the system can be
defined rigorously and produces a new element of the same O*-algebra.
Therefore, it is the physical system, i.e. $H$ itself, which is used to
construct a natural algebraic and  topological framework. This
procedure, however, is based on the very strong assumption that $H$
exists, assumption which is quite often false in many $QM_\infty$
systems, i.e. in many quantum systems with infinite degrees of
freedom. It is enough to think to the mean field spin models, for
which only the finite volume hamiltonian $H_V$ makes sense, and no
limit of $H_V$ does exist at all in any reasonable topology. For this
reason we have also discussed, along the years, other possibilities for defining $\alpha^t$
which have produced some results and the feeling that other and more general statements could be
proved. This is indeed the main motivation of this paper: our final aim is to construct an algebraic framework whose definition is
completely independent from the physical system we want to describe
or, at least, independent up to a certain extent. For that we first recall the following very general settings for $QM_\infty$ systems, which is very well described in \cite{sew1,sew2} and which has been adapted to $O^*-$algebras recently, \cite{bit2,bagrev}. The full description of a
physical system $\Sc$ implies the knowledge of three basic
ingredients: the set of the observables, the set of the states and,
finally, the dynamics that describes the time evolution of the
system by means of the time dependence of the expectation value of a
given observable on a given state. Originally the set of the
observables was considered to be a C*-algebra, \cite{hk}. In many
applications, however, this was shown not to be the most convenient
choice and the C*-algebra was replaced by a von Neumann algebra,
because the role of the representations turns out to be crucial
mainly when long range interactions are involved. Here we use a
different algebraic structure: because of the relevance of the
unbounded operators in the description of $\Sc$, we will assume that
the observables of the system belong to a quasi *-algebra
$(\A,\Ao)$, see \cite{ctrev} and references therein. The set of
states over $(\A,\Ao)$, $\Sigma$, is described again in
\cite{ctrev}, while the dynamics is usually a group (or a semigroup)
of automorphisms of the algebra, $\alpha^t$. Therefore, following
\cite{sew1,sew2}, we simply write $\Sc=\{(\A,\Ao),\Sigma,\alpha^t\}$.

The system $\Sc$ is now {\em regularized}: we introduce some cutoff
$L$, (e.g. a volume or an occupation number cutoff), belonging to a
certain set $\Lambda$, so that $\Sc$ is replaced by a sequence or,
more generally, a net of systems $\Sc_L$, one for each value of
$L\in\Lambda$. This cutoff is chosen in such a way that all the
observables of $\Sc_L$ belong to a certain
*-algebra $\A_L$ contained in $\Ao$: $\A_L\subset\Ao\subset\A$. As
for the states, we choose $\Sigma_L=\Sigma$, that is, the set of
states over $\A_L$ is taken to coincide with the set of states over
$\A$. This is a common choice, \cite{bm}, even if also different
possibilities are considered in literature. For instance, in
\cite{sewbag},  the states  depend on the cut-off $L$. Finally, since the
dynamics is related to a hamiltonian operator $H$, if this exists, and since $H$ has
to be replaced with $\{H_L\}$, $\alpha^t$ is
replaced by the family $\alpha_L^t(\cdot)=e^{iH_Lt}\cdot
e^{-iH_Lt}$. Therefore
$$
\Sc=\{(\A,\Ao),\Sigma,\alpha^t\}\longrightarrow\{\Sc_L=\{\A_L,\Sigma,\alpha_L^t\},L\in\Lambda\}.
$$

This is the general settings in which we are going to work. In Section II we will assume that the physical topological quasi *-algebra is somehow related to the family of regularized hamiltonians, $\{H_L\}$. In Section III we will remove this assumption, paying some price for this attempt of generalization. Section IV contains some physical applications while the conclusions are contained in Section V.

We devote the rest of this section to introduce few useful notation
on  quasi
*-algebras, which will be used in the rest of this paper.

\vspace{3mm}

Let $\A$ be a linear space, $\Ao\subset\A$ a $^\ast$-algebra with
unit $\1$ (otherwise  we can always add it): $\A$ is a {\it   quasi
$^\ast$-algebra over $\Ao$} if

\vspace{3mm}

{\bf [i]} the right and left multiplications of an element of $\A$
and an element of $\Ao$ are always defined and linear;

%\vskip

{\bf [ii]} $x_1 (x_2 a)= (x_1x_2 )a, (ax_1)x_2= a(x_1 x_2)$ and
$x_1(a
 x_2)= (x_1 a) x_2$, for each $x_1, x_2 \in \A_0$ and $a \in \A$;

%\vskip

{\bf [iii]} an involution * (which extends the involution of $\Ao$)
is defined in $\A$ with the property $(ab) ^\ast =b ^\ast a ^\ast$
whenever the multiplication is defined.

%\vskip

%We will always assume that our quasi $^\ast$ -algebra
%{\it  has a unit}, i.e. an element $\1 \in \Ao$ such that
%$a\1 =\1 a=a, \;\, \forall a\in \A$.

\vspace{3mm}

A quasi  $^\ast$ -algebra $(\A,\Ao)$ is {\it   locally convex} (or
{\it   topological}) if in $\A$ a { locally convex topology} $\hat\tau$
is defined such that (a) the involution is continuous and the
multiplications are separately continuous; and (b) $\Ao$ is dense in
$\A[\hat\tau]$.

Let $\{p_\alpha\}$ be a directed set of seminorms which defines
$\hat\tau$. The existence of such a directed set can always be assumed.
We can further  also assume that $\A[\hat\tau]$ is {\it complete}.
Indeed, if this is not so, then the $\hat\tau$-completion
$\tilde\A[\hat\tau]$ is again a topological quasi *-algebra over the
same
*-algebra $\Ao$.

\vspace{2mm}

A relevant example of a locally convex quasi *-algebra of operators can be
constructed as follows: let $\Hil$ be a separable Hilbert space and
$N$ an unbounded, densely defined, self-adjoint operator. Let
$D(N^k)$ be the domain of the operator $N^k$, $k\in \Bbb{N}_0=\Bbb{N}\cup\{0\}$, and $\D$
the domain of all the powers of $N$:  $ \D \equiv D^\infty(N) =
\cap_{k\geq 0} D(N^k). $ This set is dense in $\Hil$. Let us now
introduce $\Lc^\dagger(\D)$, the
*-algebra of all the {  closable operators} defined on $\D$ which,
together with their adjoints, map $\D$ into itself. Here the adjoint
of $X\in\Lc^\dagger(\D)$ is { $X^\dagger=X^*{\restr \D}$}.

\vspace{3mm}

In $\D$ the topology is defined by the following $N$-depending
seminorms: $\phi \in \D \mapsto \|\phi\|_n\equiv \|N^n\phi\|,$
 $n\in \mathbb{N}_0$, while the topology $\tau$ in $\Lc^\dagger(\D)$ is introduced by the seminorms
{\normalsize$$ X\in \Lc^\dagger(\D) \mapsto \|X\|^{f,k} \equiv
\max\left\{\|f(N)XN^k\|,\|N^kXf(N)\|\right\},\vspace{-2mm}$$} where
$k\in\mathbb{N}_0$ and   $f\in\C$, the set of all the positive,
bounded and continuous functions  on $\mathbb{R}_+$, which decrease
faster than any inverse power of $x$: $\Lc^\dagger(\D)[\tau]$ is a
{   complete *-algebra}.

It is clear that $\Lc^\dagger(\D)$ contains unbounded operators.
Indeed, just to consider the easiest examples, it contains all the
positive powers of $N$. Moreover, if for instance $N$ is the closure
of $N_o=a^\dagger\,a$, with $[a,a^\dagger]=\1$, $\Lc^\dagger(\D)$
also contains all positive powers of $a$ and $a^\dagger$.
    \vspace{2mm}

Let further {  $\Lc(\D,\D')$} be the set of all continuous maps from
$\D$ into $\D'$, with their topologies, \cite{aitbook}, and let
$\hat\tau$ denotes the topology defined by the seminorms $$ X\in
\Lc(\D,\D') \mapsto \|X\|^{f} = \|f(N)Xf(N)\|,$$ $f\in\C$. Then
$\Lc(\D,\D')[\hat\tau]$ is a { complete vector space}.

In this case { $\Lc^\dagger(\D)\subset\Lc(\D,\D')$} and the pair
$$(\Lc(\D_,\D')[\hat\tau],\Lc^\dagger(\D)[\tau])$$ is a {\it concrete realization} of a locally convex quasi
*-algebra.

\vspace{2mm}

Other examples of algebras of unbounded operators have been
introduced along the years, but since these will play no role in
this paper, we simply refer to \cite{aitbook} and \cite{bagrev} for
further results.

\section{A step toward generalization}

In this section we will generalize
our previous results, \cite{bagrev,bt3,bt4}, in the attempt to build
up a general algebraic setting which works well independently of the
particular physical system we are considering. In particular in
this section we will assume  that the finite volume hamiltonian $H_L$ associated to
the regularized system $\Sc_L$ mutually commute:
$[H_{L_1},H_{L_2}]=0$ for all $L_1, L_2\subset \Lambda$. It is
worth noticing that this requirement is not satisfied in general
for, e.g., mean field spin models while it holds true when we adopt the procedure
discussed for instance in \cite{bt4}. In this case $H_L$ is deduced
by an existing $H$ simply adopting an {\em occupation number} cutoff and the different $H_L$'s mutually commute.

\vspace{2mm}

Let $\{P_l, \,l\geq 0\}$ be a sequence of orthogonal projectors:
$P_l=P_l^\dagger=P_l^2$, $\forall\,l\geq 0$, with $P_lP_k=0$ if $l\neq k$ and $\sum_{l=0}^\infty P_l=\1$. These operators are
assumed to give the spectral decomposition of a certain operator
$S=\sum_{l=0}^\infty s_l\,P_l$. The coefficients $s_l$  are all non
negative and, in order to make the situation more interesting, they
diverge monotonically to $+\infty$ as $l$ goes to infinity. Also, here and in the rest of
the paper we will assume that $S$ is invertible and that $S^{-1}$ is a
bounded operator. Hence there exists $s>0$ such that $s_l\geq s$ for all $l\geq 0$. This is quite often a reasonable assumption which
does not change the essence of the problem at least if $S$ is
bounded from below. Since $P_lP_s=\delta_{l,s}P_l$ the operator
$Q_L:=\sum_{l=0}^LP_l$ is again a projection operator satisfying
$Q_L=Q_L^\dagger$ as well as $Q_LQ_M=Q_{min(L,M)}$. It is easy to
check that each vector of the set
$\E=\{\varphi_L:=Q_L\varphi,\,\varphi\in\Hil,\, L\geq 0\}$ belongs to the domain of $S$, $D(S)$, which is
therefore dense in $\Hil$. Hence, the operator $S$ is
self-adjoint, unbounded and densely-defined, and can be used as in the previous
section to define the algebraic and topological framework we will
work with. Let $\D=D^\infty(S):=\cap_{k\geq 0}D(S^k)$ be the domain
of all the powers of $S$. Since $\E\subseteq\D$, this set is also dense in
$\Hil$. Following the example discussed in the previous section, we
introduce now $\Lc^\dagger(\D)$ and the topology $\tau$ defined by
the following seminorms: $$ X\in \Lc^\dagger(\D) \mapsto
\|X\|^{f,k} \equiv \max\left\{\|f(S)XS^k\|,\|S^kXf(S)\|\right\},$$
where $k\in\mathbb{N}_0$ and   $f\in\C$. Briefly we will say that
$(f,k)\in\C_0:=(\C,\Bbb{N}_0)$. $\Lc^\dagger(\D)[\tau]$ is a
complete
*-algebra which is the $O^*$-algebra we will use. In the rest of the paper, for simplicity, we
will identify $\|X\|^{f,k}$ simply with $\|f(S)XS^k\|$. The
estimates for $\|S^kXf(S)\|$ are completely analogous and are left
to the reader.

As in Section I we could also introduce {  $\Lc(\D,\D')$}, the set of all continuous maps from
$\D$ into $\D'$, with their topologies, and  the topology $\hat\tau$ on it
defined by the seminorms $$ X\in \Lc(\D,\D') \mapsto \|X\|^{f} =
\|f(S)Xf(S)\|,$$ $f\in\C$. Then $\Lc(\D,\D')[\hat\tau]$ is a {
complete vector space}. Of course $(\Lc(\D_,\D')[\hat\tau],\Lc^\dagger(\D)[\tau])$ is a  locally convex quasi
*-algebra.

\vspace{3mm}

Let us now introduce a family of bounded operators
$\{H_M=\sum_{l=0}^Mh_l\,P_l\}$, $M\geq 0$, where $h_l$ are real
numbers. These operators satisfy the following: \be
H_M=H_M^\dagger,\quad \|H_M\|\leq \sqrt{\sum_{l=0}^Mh_l^2},\quad
[H_L,H_M]=0, \label{31}\en for all $L,M\geq 0$\footnote{Notice that we could also have followed a {\em reverse} approach: we use $\{H_L\}$ to define, if possible, a complete set of orthogonal projectors $\{P_l\}$, $\sum_lP_l=\1$, and, using these operators, a self-adjoint, unbounded and densely-defined operator $S$. Finally we use $S$ to construct the $O^*$-algebra $\LD$ and the topology $\tau$.}. We also have, for all $\varphi\in\D$ and $\forall\,L$, $SH_L\varphi=H_LS\varphi$. In the following we will simply say that $H_L$ and $S$ commute: $[S,H_L]=0$. For the time being,
we do not impose any other requirement on $\{h_l\}$. On the
contrary, since we are interested in {\em generalizing the
procedure}, we want our $\{h_l\}$ to produce situations apparently
out of control. For this reason we are interested in considering $h_l$
very rapidly increasing with $l$. For instance, if $f_0(x)$ is a
fixed function in $\C$ and if $h_l=\left(f_0(s_l)\right)^{-1}$, then
it is an easy exercise to check that the sequence $\{H_L\}$ does not
converge in the topologies $\tau$ or
$\hat\tau$. For this reason, and with this choice for $h_l$, the
sequence $H_L$ does not converge  to an operator $H$ of
$\Lc^\dagger(\D)$ or  to an element of $\Lc(\D,\D')$. Another
simple example can be constructed taking
$h_l=\left(f_0(s_l)\right)^{-2}$. Once again $\{H_L\}$ does not
converge  in the topology $\hat\tau$, and therefore it cannot define
an element of $\Lc(\D,\D')$. In other words, there exist conditions
on $h_l$ which prevents the sequence $\{H_L\}$ to define an element
of the topological quasi *-algebra defined by $S$. On the contrary,
if $h_l$ goes to infinity { but not too fast, for instance as some inverse power of $l$}, then
$\tau-\lim_{L,\infty}H_L$ exists and defines an element of
$\Lc^\dagger(\D)$. Moreover, if $\{h_l\}\in l^1(\Bbb{N})$, the limit of $H_L$ exists and belongs to
$B(\Hil)$.

\vspace{2mm}

The situation we are interested in is {\em the ugly one}: the
algebraic framework is fixed by $S$ while no hamiltonian operator
exists for the physical system $\Sc$. Nevertheless we will see that
even under these assumptions the algebraic dynamics of $\Sc$ can be
defined. More in details, the following proposition holds true:

\begin{prop}
Suppose that for some  $n\geq1$ $\{s_l^{-n}\}\in l^2(\Bbb{N})$. Then
we have
\begin{enumerate}
\item $\forall\,t\in\Bbb{R}$ \be
\lim_{L,M\rightarrow\infty}\left\|S^{-n}\left(e^{iH_Lt}-e^{iH_Mt}\right)\right\|=0;\label{32}\en

\item $\forall\,t\in\Bbb{R}$  $\tau-\lim_L\,e^{iH_Lt}=:T_t$ exists in $\Lc^\dagger(\D)$;

\item $\forall X\in\Lc^\dagger(\D)$ and $\forall\,t\in\Bbb{R}$, $\tau-\lim_L\,e^{iH_Lt}\,X\,e^{-iH_Lt}=:\alpha^t(X)$ exists in
$\Lc^\dagger(\D)$;

\item $\forall X\in\Lc^\dagger(\D)$, $\alpha^t(X)=T_tXT_{-t}$, $\forall\,t\in\Bbb{R}$.

\end{enumerate}

\label{prop21}
\end{prop}

\begin{proof}

\begin{enumerate}

\item

First of all let us recall that, for all $L$ and $M$, $[H_L,H_M]=0$ and $[S,H_L]=0$. Therefore, assuming that $M>L$ to fix the
ideas, and calling $H_{M,L}:=H_M-H_L$, with simple computations we
have
$$
\left\|S^{-n}\left(e^{iH_Lt}-e^{iH_Mt}\right)\right\|=2
\left\|S^{-n}\sin\left(\frac{tH_{M,L}}{2}\right)\right\|=$$
$$=\left\|\sum_{k=L+1}^M\frac{1}{s_k^n}\,
\sin\left(\frac{th_k}{2}\right)\,P_k\right\|\leq
\sqrt{\sum_{k=L+1}^M\frac{1}{s_k^{2n}}}\rightarrow 0,
$$
when $L,M\rightarrow\infty$ because of our assumption on the
sequence $\{s_l\}$.

\item The second statement follows from the previous result and from the following simple
estimate:
$$
\left\|f(S)\left(e^{iH_Lt}-e^{iH_Mt}\right)S^k\right\|\leq
\|f(S)S^{k+n}\|\,
\left\|S^{-n}\left(e^{iH_Lt}-e^{iH_Mt}\right)\right\|,
$$
where $n$ is the positive integer in our assumption. This is a consequence of the commutativity between $H_L$ and $S$, and  clearly implies that the left-hand side goes
to zero when $L,M\rightarrow\infty$. We call $T_t$ the limit of the
sequence $\{e^{iH_Lt}\}$  in $\tau$. Needless to say, since $\Lc^\dagger(\D)$ is
$\tau$-complete, $T_t\in \Lc^\dagger(\D)$.

\item We use the following estimate, which follows from the commutativity between $H_L$ and $S$ and of the fact that $e^{iH_Lt}$ is unitary:
$$
\left\|f(S)\,\left(e^{iH_Lt}\,X\,e^{-iH_Lt}-e^{iH_Mt}\,X\,e^{-iH_Mt}\right)S^k\right\|\leq$$
$$\leq
\left\|f(S)e^{iH_Lt}X\left(e^{-iH_Lt}-e^{-iH_Mt}\right)S^k\right\|+\left\|f(S)\left(e^{iH_Lt}-e^{iH_Mt}\right)Xe^{-iH_Mt}S^k\right\|\leq
$$
$$
\leq
\left\|f(S)e^{iH_Lt}XS^{k+n}\right\|\,\left\|S^{-k-n}\left(e^{-iH_Lt}-e^{-iH_Mt}\right)S^k\right\|+
\left\|f(S)\left(e^{iH_Lt}-e^{iH_Mt}\right)XS^k\right\|\rightarrow0
$$
when $L,M$ go to infinity because of the previous results and
of the separate continuity of the multiplication in $\tau$. Of course, using
again the completeness of $\Lc^\dagger(\D)$, this means that for
each $X\in \Lc^\dagger(\D)$  and for each $t\in\Bbb{R}$, there
exists an element of $\Lc^\dagger(\D)$ which we call $\alpha^t(X)$,
which is the $\tau$-limit of the regularized time evolution
$\alpha_L^t(X):=e^{iH_Lt}\,X\,e^{-iH_Lt}$.

\item This last statement shows that $\alpha^t(X)$ can also be
obtained in a different way, just referring to the operator $T_t$
defined in 2. of this proposition. In other word, for each $X\in
\Lc^\dagger(\D)$ we have \be \alpha^t(X)=\tau-\lim_L
e^{iH_Lt}\,X\,e^{-iH_Lt}=\left(\tau-\lim_L e^{iH_Lt}\right)\,X\,
\left(\tau-\lim_L e^{-iH_Lt}\right)\label{33}\en The proof of this
equality goes like this: let us take $X\in \Lc^\dagger(\D)$. Then we
have
$$
\|\alpha^t(X)-T_tXT_{-t}\|^{f,k}\leq
\|\alpha^t(X)-\alpha_L^t(X)\|^{f,k}
+\|\alpha_L^t(X)-e^{iH_Lt}\,XT_{-t})\|^{f,k}+$$
$$+\|e^{iH_Lt}\,XT_{-t}-T_tXT_{-t}\|^{f,k},$$
and this right-hand side goes to zero term by term because of the previous results and
of the separate continuity of the multiplication in the topology
$\tau$. This concludes the proof.
\end{enumerate}
\end{proof}

\vspace{2mm}

{\bf Remark:} It is worth stressing here that the above proposition
only gives sufficient conditions for an algebraic dynamics to be
defined. In fact, we expect that milder conditions could suffice if we analyze directly the sequence $\{\alpha_L^t(X)\}$, instead of
$\{e^{iH_Lt}\}$ as we have done here. This is because $\alpha_L^t(X)$
contains commutators like $[H_L,X]$ which can be analyzed very
simply if $X$ is sharply localized, for instance, in some particular
lattice site and $H_L$ is localized in a finite volume labeled by
$L$. This is what happens, for instance, in spin models. We will consider this point of view in a further paper.

\vspace{2mm}

It is also possible to introduce $\alpha^t(X)$ starting from the
infinitesimal dynamics, i.e. from the derivation. For that we
introduce the following subset of $\Lc^\dagger(\D)$,
$\Lc_0^\dagger(\D):=\{x_M:=Q_MxQ_M, \, M\geq 0,\, x\in
\Lc^\dagger(\D)\}$. If $\{s_l^{-1}\}$ belongs to $l^2(\Bbb{N})$ this
set is $\tau$-dense in $\Lc^\dagger(\D)$. Indeed, let us take $y\in
\Lc^\dagger(\D)$ and let us put $y_M=Q_MyQ_M$. Then we have
$$
\|y-y_M\|^{f,k}\leq \|(\1-Q_M)S^{-1}\|\,\|f(S)SyS^k\|+
$$
$$
+\|f(S)yS^{k+1}\|\,\|S^{-1}(\1-Q_M)\|\leq
\sqrt{\sum_{l=M+1}^\infty\,s_l^{-2}}\,\left(\|f(S)SyS^k\|+\|f(S)yS^{k+1}\|\right)\rightarrow
0
$$
for $M\rightarrow \infty$. In our assumptions it is possible to
check directly that, for each $X\in \Lc^\dagger(\D)$ and for each fixed $L$,
\be
\tau-\lim_{M\rightarrow\infty}\,\alpha_L^t(X_M)=\alpha_L^t(X),\qquad
\tau-\lim_{M\rightarrow\infty}\,\alpha^t(X_M)=\alpha^t(X)
\label{34}\en Moreover, if we introduce the {\em regularized
derivation} as $\delta_L(X)=i[H_L,X]$ and, by recursion,
$\delta_L^k(X)=i[H_L,\delta_L^{k-1}(X)]$, $k\geq 1$, we also find
\be \delta_L^k(X_M)=\delta_M^k(X_M)=Q_M\delta_M^k(X_M)Q_M=
Q_M\delta_M^k(X)Q_M, \label{35} \en for each $X\in \Lc^\dagger(\D)$,
for each $k\in\Bbb{N}$ and for each $L\geq M$. Again, the proof of
this statement, which is easily deduced by induction, is left to the
reader. We just want to notice here that the proof is strongly based on the relations between $S$, $H_L$ and $Q_L$.

Another useful result is given in the following equality: \be
\alpha_L^t(X_M)=Q_M\alpha_L^t(X_M)Q_M=Q_M\alpha_L^t(X)Q_M,
\label{36}\en which again holds for each $X\in \Lc^\dagger(\D)$ and
for each $L$ and $M$. This is a direct consequence of the following
commutation rule: $[H_L,Q_M]=0$, $\forall\,L,M$. As a consequence of
(\ref{36})  we deduce that, as already stated in  (\ref{34}), for each
fixed $L$, $\tau-\lim_{M}\,\alpha_L^t(X_M)=\alpha_L^t(X)$. In
order to relate the infinitesimal and the finite dynamics, we still
need another result. For each $X\in\Lc^\dagger(\D)$ , for each $L,
M, k$ and for each $f(x)\in\C$ we put $X^{(f,k)}_M:=f(S)X_MS^k$. Of
course, since $[S,Q_M]=0$ we can write
$X^{(f,k)}_M=Q_M\left(f(S)XS^k\right)Q_M=Q_MX^{(f,k)}Q_M$. Then we
have: (1) $X^{(f,k)}\in\Lc^\dagger(\D)$; (2)
$f(S)\alpha_L^t(X_M)S^k=\alpha_L^t(X_M^{(f,k)})$; (3)
$f(S)\delta_M^l(X_M)S^k=\delta_M^l(X_M^{(f,k)})$, $\forall l\geq 1$.
The first two assertions are trivial. The proof of the last
one is a bit more difficult and, again, can be deduced by induction
on $l$. This statement is useful to prove that, $\forall
X\in\Lc^\dagger(\D)$ and for all $L,M$, \be
\tau-\lim_{N,\infty}\,\sum_{k=0}^N
\,\frac{t^k}{k!}\,\delta_L^k(X_M)=\alpha_L^t(X_M)\label{37}\en Let
us prove this statement. Since for each fixed $L$ and $M$ both $H_L$ and  $X_M$
are bounded operators, we have
$$
\left\|\alpha_L^t(X_M)-\sum_{k=0}^N\frac{t^k}{k!}\,\delta_L^k(X_M)\right\|=
\left\|e^{iH_Lt}X_Me^{-iH_Lt}-\sum_{k=0}^N\frac{t^k}{k!}\,\delta_L^k(X_M)\right\|\rightarrow
0
$$
when $N\rightarrow\infty$. Therefore, for each
$(f,k)\in\C_0$,
$$
\left\|\alpha_L^t(X_M)-\sum_{k=0}^N\frac{t^k}{k!}\,\delta_L^k(X_M)\right\|^{f,k}=
\left\|\alpha_L^t(X_M^{(f,k)})-\sum_{k=0}^N\frac{t^k}{k!}\,\delta_L^k(X_M^{(f,k)})\right\|,
$$
which again goes to zero when $N$ goes to infinity. Therefore, using (\ref{34}),
we conclude that $$
\alpha^t(X)=\tau-\lim_L\,\alpha_L^t(X)=\tau-\lim_L\left(\tau-\lim_M\,\alpha_L^t(X_M)\right)=$$
\be
=\tau-\lim_{L,M,N\rightarrow\infty}\sum_{k=0}^N\frac{t^k}{k!}\delta_L^k(X_M)\label{38}\en
This result shows the relation between {\em derivation} and {\em
time evolution}, giving still another possibility for defining the
time evolution in $\Lc^\dagger(\D)$. Notice that here the order in
which the limits are taken is important.

We refer to Section IV for some applications of our results to
$QM_\infty$. More applications will be discussed in a paper in
preparation.

\vspace{2mm}

{\bf Remarks:} (1) Not many substantial differences arise when the
operators $S$ and $H_L$ have continuous (or mixed) spectra. In this
case, instead of the {\em discrete formulas} given above, we must
use $S=\int_{0}^\infty s(\lambda)\,dE_\lambda$ and $H_L=\int_{0}^L
h(\lambda)\,dE_\lambda$.

(2) Secondly it is worth stressing that this approach naturally
extend our previous results, see \cite{bt4} for instance, since if
$h_l\equiv s_l$ we recover what is stated in that paper. We believe
that the present situation is more relevant for applications to
$QM_\infty$ since it represents a first step in applying quasi
*-algebras for the analysis of physical models which {\bf do not}
admit a global hamiltonian.

\vspace{2mm}

\subsection{Some other considerations}

Up to now we have made a strong assumption, i.e. that $S$ and $H_L$
admit the same spectral projections. This is something we would
avoid, if possible, since it is rarely true in real
quantum mechanical models. We still suppose that $S$ and $H_L$ have discrete
spectra, $S=\sum_{l=0}^\infty s_lP_l$ and $H_M=\sum_{l=0}^M h_lE_l$,
but we also consider here the case in which $E_j$ and $P_j$ do not coincide for all $j\in\Bbb{N}_0$. Then almost all the same conclusions as before
can still be deduced if, for instance, $P_j=E_j$ definitively, i.e.
for $j\geq M$ for some fixed $M$, or if
$[E_l,P_j]=0$ for all $l,j$.

A less trivial condition is the following one: let $\{\varphi_l\}$ and
$\{\psi_l\}$ be two different orthonormal bases of $\Hil$ and
suppose that $P_l=|\varphi_l><\varphi_l|$ and
$E_l=|\psi_l><\psi_l|$. Here we are using the Dirac bra-ket notation. It is clear that $[E_l,P_j]\neq0$ in
general. Nevertheless, if the $\psi_l$'s are finite linear
combinations of the $\varphi_j$'s, then again the above results
still can be proved. We do not give here the detailed proof of these
claims, since they do not differ very much from those given before
but for some extra difficulties arising here and there in the various
estimates. We will prove an analogous result in the next section, starting from quite a different assumption.

We end this section with a final comment concerning the limit of Gibbs states in the
conditions considered so far.
 Let
$\omega_L(.)$ be the linear functional defined on
$(\Lc(\D,\D'),\Lc^\dagger(\D))$ as $\omega_L(.)=\frac{tr(e^{-\beta
H_L}.)}{N_L}$, $N_L=tr(e^{-\beta H_L})$. Here $tr$ is defined as
$tr(A)=\sum_{k=0}^\infty <f_k,Af_k>$, $\{f_k\}$ being an o.n. basis
of $\Hil$ contained in $\D$ and $\beta>0$. Then it is clear that
$\omega_L(\1)=1$ for all $L$ and it is not hard to check that, for
all $(f,k)\in\C_0$,
$$
\left\|f(S)\left(\frac{e^{-\beta H_M}}{N_M}-\frac{e^{-\beta
H_L}}{N_L}\right)S^k\right\|\rightarrow 0
$$
when $L,M\rightarrow\infty$.

\vspace{3mm}

This means that calling  $\rho_L:=\frac{e^{-\beta
H_L}}{tr_L\left(e^{-\beta H_L}\right)}$ the density matrix of a
Gibbs state at the inverse temperature $\beta$, then
$\tau-\lim_L\rho_L$ exists in $\LD$. But it is still to be
investigated whether this limit is a KMS state (in some {\em
physical sense}). We refer to \cite{aitbook} for some remarks
on KMS states on partial *-algebras.

\section{A more general setting}

In the previous section we have discussed what we have called {\em a
first step toward generalization}: indeed, the main improvement with
respect to our older results is mainly the (crucial) fact that
the sequence of {\em regularized hamiltonians} $H_L$ do not necessarily converge
in the $\tau$-topology to an element of $\LD$. It also does not
converge in $\Lc(\D,\D')$. So we have no hamiltonian for $\Sc$ but
only a regularized family of subsystems $\Sc_L$, where $L$ is the
regularizing cutoff, and their related regularized hamiltonians $H_L$. However we have worked with the following strong
requirement: $[H_L,H_M]=0$ for all $L, M$ and, moreover, for each
$L$ the hamiltonian $H_L$ commutes (in the sense discussed previously) with
$S$. In this section we will work without these assumptions,
generalizing our results as much as we can to a much more general
situation: once again $\{H_L\}$ does not converge in any suitable
topology. Moreover, no a-priori relation between $H_L$ and $S$ is assumed and the different $H_L$'s are not required to commute in general. This, we believe,
makes our next results more useful for general models in $QM_\infty$. However, since in general there is no reason for $e^{iH_Lt}$, $L\geq0$, to leave $\D$ invariant, see \cite{timm} for counterexamples, we simply require that for all $L$ the operator $e^{iH_Lt}$ belongs to $\LD$.

We begin by stating the main proposition of this section which
extends Proposition \ref{prop21} to this more general setting. In particular we will give a sufficient condition which allows us to define the algebraic
dynamics $\alpha^t$. In the second part of this section  we will see
if and when this condition is satisfied. We give here the
details of the proof, even if they may look very similar to those
of Proposition \ref{prop21}, in order to stress where the
commutativity between $H_L$ and $S$ plays a role and where it does not.

\begin{prop}\label{prop22}
Suppose that $\forall\, k\geq0$ $\exists\,n> 0$ such that \be
\lim_{L,M\rightarrow\infty}\left\|S^{-k-n}\left(e^{iH_Lt}-e^{iH_Mt}\right)S^k\right\|=0
\label{41}\en for $t\in\Bbb{R}$. Then:
\begin{enumerate}
\item there exists $T_t=\tau-\lim_L\,e^{iH_Lt}\in\LD$;
\item for all $X\in\LD$ and for $t\in\Bbb{R}$ there exists $\alpha^t(X)=
\tau-\lim_L\,e^{iH_Lt}Xe^{-iH_Lt}\in\LD$;
\item for all $X\in\LD$ and for $t\in\Bbb{R}$ we have
$\alpha^t(X)=T_tXT_{-t}$.
\end{enumerate}

\end{prop}
\begin{proof}
\begin{enumerate}
\item
Let $(f,k)\in\C_0$ and let $n>0$ be the  integer fixed in our assumptions. Then we have
$$
\left\|f(S)\left(e^{iH_Lt}-e^{iH_Mt}\right)S^k\right\|\leq
\left\|f(S)\1S^{k+n}\right\|\left\|S^{-k-n}\left(e^{iH_Lt}-e^{iH_Mt}\right)S^k\right\|\rightarrow0
$$
when $L,M\rightarrow\infty$, since $\1\in\LD$. Recalling that $\LD$ is
$\tau-$complete, there exists $T_t\in\LD$ which is the $\tau-$limit
of $\left\{e^{iH_Lt}\right\}$.
\item
To prove our statement, we observe that
$$
\left\|f(S)\,\left(e^{iH_Lt}\,X\,e^{-iH_Lt}-e^{iH_Mt}\,X\,e^{-iH_Mt}\right)S^k\right\|\leq$$
$$\leq
\|f(S)e^{iH_Lt}XS^{k+n}\|\,\left\|S^{-k-n}\left(e^{-iH_Lt}-e^{-iH_Mt}\right)S^k\right\|+$$
$$+\left\|f(S)\left(e^{iH_Lt}-e^{iH_Mt}\right)XS^{k+n}\right\|\,\|S^{-k-n}e^{-iH_Mt}S^k\|
$$
Notice that this result looks slightly different from the analogous
one in Proposition \ref{prop21} because no commutativity is assumed
here. The right hand side above goes to zero when $L,M$ go to infinity. In fact, since the
multiplication is separately continuous and since $e^{-iH_Lt}$ is
$\tau-$converging, then $e^{-iH_Lt}X$ $\tau-$converges as well.
Therefore, for each $L$, $\|f(S)e^{iH_Lt}XS^{k+n}\|$ is bounded, and
bounded is also $\|S^{-k-n}e^{-iH_Mt}S^k\|$ for each $M$ and $t$, due to our
assumption. Finally, using  the completeness of $\Lc^\dagger(\D)$,
 there exists an element of $\Lc^\dagger(\D)$ which we
call $\alpha^t(X)$, which is the $\tau$-limit of
$\alpha_L^t(X):=e^{iH_Lt}\,X\,e^{-iH_Lt}$.

\item Again, the statement is close to that in Proposition \ref{prop21}, but the proof is slightly different since
$[H_L,S]\neq 0$. For each $X\in \Lc^\dagger(\D)$, $\forall\,t\in\Bbb{R}$ and for each
$(f,k)\in\C_0$, we have
$$
\|\alpha_L^t(X)-T_tXT_{-t}\|^{f,k}\leq\|\alpha_L^t(X)-T_tXe^{-iH_Lt}\|^{f,k}
+\|T_tXe^{-iH_Lt}-T_t\,XT_{-t})\|^{f,k}\leq$$
$$\leq \|f(S)\left(e^{iH_Lt}-T_t\right)XS^{k+n}\|\,\|S^{-k-n}e^{-iH_Lt}S^k\|+
\|f(S)T_tX\left(e^{-iH_Lt}-T_{-t}\right)S^{k}\|$$ which goes to zero
term by term because of the previous results and of the separate
continuity of the multiplication in the topology $\tau$. This implies the statement.
\end{enumerate}
\end{proof}

\vspace{2mm}

We are now ready to look for conditions, easily verifiable,
which imply the main hypothesis of this proposition, i.e. the existence,
for each $k\geq 0$, of a positive integer $n$ such that
$\lim_{L,M\rightarrow\infty}\left\|S^{-k-n}\left(e^{iH_Lt}-e^{iH_Mt}\right)S^k\right\|=0$ for $t\in\Bbb{R}$.

The first trivial remark is the following: if $[H_L,S^{-1}]=0$ for
all $L$ then the above condition simplifies. In this case, it is
enough to prove that there exists $n>0$ such that
$\lim_{L,M\rightarrow\infty}\left\|S^{-n}\left(e^{iH_Lt}-e^{iH_Mt}\right)\right\|=0$.
This requirement is now completely analogous to that of the previous
section, and is satisfied if the $s_l$ in the spectral
decomposition of $S$ are such that $\{s_l^{-n}\}$ belongs to
$l^2(\Bbb{N})$.

In this section we are more interested in considering the situation
where $[H_L,S^{-1}]\neq0$, and so it is not surprising that we will
not be able to define an algebraic dynamics in all of $\LD$ but only on
certain subspaces. First, for $\alpha>0$, we define the following
set: \be \A_S^{(\alpha)}=\{X\in B(\Hil)\cap\LD,\quad \|T_S(X)\|\leq
\alpha \|X\|\}\label{42}\en where $T_S(X)=S^{-1}XS$. Of course, the
first thing to do is to check that $\A_S^{(\alpha)}$ is not empty
and, even more, that it contains sufficiently many elements. It is
indeed clear that all the multiples of the identity operator $\1$
belong to $\A_S^{(1)}$. It is clear as well that by the definition
itself, $\A_S^{(\alpha)}$ consists of bounded operators which are
still bounded after the action of $T_S$. For instance all the
multiples and the positive powers of $S^{-1}$ belong to some
$\A_S^{(\alpha)}$, as well as all those bounded operators which
commute with $S^{-1}$. Another trivial feature is that the set
$\{\A_S^{(\alpha)}, \,\alpha>0\}$ is a chain: if
$\alpha_1\leq\alpha_2\leq\alpha_3\leq\ldots$ then
$\A_S^{(\alpha_1)}\subseteq\A_S^{(\alpha_2)}\subseteq\A_S^{(\alpha_3)}\subseteq\cdots$.
We also notice that $\A_S^{(\alpha)}$ is not an algebra by itself,
since if $X,Y\in\A_S^{(\alpha)}$ then $XY\in\A_S^{(\alpha^2)}$.

\vspace{2mm}

{\bf Example 1:} Other elements of some $\A_S^{(\alpha)}$ are those
$X\in B(\Hil)\cap\LD$ whose commutator (in the sense of the
unbounded operators) with $S$ looks like $[X,S]=BX$, for some
bounded operator $B$. In this case, in fact, it is easily checked
that $X\in\A_S^{(\hat\alpha)}$, where $\hat\alpha=\|\1+S^{-1}B\|$.

Before going on it is worth noticing that this requirement is
satisfied in a concrete physical system. For instance, let us consider an
infinite $d-$dimensional lattice $\Lambda$. To the lattice site $j$
we can {\em attach} a two-dimensional Hilbert space $\Hil_j$ spanned
by $\varphi_0^{(j)}$ and
$\varphi_1^{(j)}:=a_j^\dagger\varphi_0^{(j)}$, where $a_j$ and
$a_j^\dagger$ are the fermionic operators satisfying
$\{a_j,a_k^\dagger\}=a_ja_k^\dagger+a_k^\dagger a_j=\delta_{j,k}$,
$j,k\in\Lambda$. $a_j$ and $a_j^\dagger$ are the so called
annihilation and creation operators, while $N_j=a_j^\dagger a_j$ is
called the number operator. It is well known that all these
operators are bounded. Let us now define the {\em total number
operator} $S=\sum_{j\in\Lambda}N_j$. Of course $S$ is unbounded on
$\Hil_\infty=\otimes_{j\in\Lambda}\,\Hil_j$, see \cite{thibook} for infinite tensor product Hilbert spaces. Now, since
operators localized in different lattice sites commute, if we take
$X=a_{j_0}$ we deduce that $[S,X]=(\1-N_{j_0})X$, which is exactly
as required above since $\1-N_{j_0}$ is bounded. This implies that
$a_{j_0}$ belongs to some $\A_S^{(\alpha)}$, as well as
$a_{j_0}^\dagger$, $N_{j_0}$ and all the {\em strictly localized}
operators, i.e. those operators $X=X_{j_0}\cdots X_{j_n}$, with
$n<\infty$.

Needless to say, this same example can be modified replacing
fermionic degrees of freedom with spin operators, or, more
generally, with $N\times N$ matrices. As for its extension to bosons, the situation is more difficult and will not be discussed here.

\vspace{2mm}

This example, although physically motivated, might suggest  the
reader  that there exist only few non trivial elements in
$\A_S^{(\alpha)}$. The following result shows that this is not
so: each $X\in\LD$ produces, in a way which we are going to
describe, another element, $X_L$, which belongs to some
$\A_S^{(\hat\alpha)}$. Moreover we will show also that the chain
$\{\A_S^{(\alpha)}, \,\alpha>0\}$ satisfies a density property in
$\LD$.

First we repeat the same construction as in the previous section:
suppose that $S=\sum_{l=0}^\infty s_l\,P_l$, where $\{P_l, \,l\geq
0\}$ is a sequence of orthogonal projectors:
$P_l=P_l^\dagger=P_l^2$, $\forall\,l\geq 0$. Moreover
$P_lP_s=\delta_{l,s}P_l$. Then the operator $Q_L:=\sum_{l=0}^LP_l$
is again a projection operator satisfying $Q_L=Q_L^\dagger$ as well
as $Q_LQ_M=Q_{min(L,M)}$. Since $SQ_L=\sum_{l=0}^L s_l\,P_l$ and
$S^{-1}Q_L=\sum_{l=0}^L s_l^{-1}\,P_l$, we get \be\|SQ_L\|\leq
\sqrt{\sum_{l=0}^L\,s_l^2}=:s_L^+,\qquad \|S^{-1}Q_L\|\leq
\sqrt{\sum_{l=0}^L\,s_l^{-2}}=:s_L^-\label{43}\en

The following Lemma can be easily proved:

\begin{lemma}\label{lemma1}
Let $X\in\LD$ and put $X_L:=Q_LXQ_L$. Then $X_L$ belongs to
$\A_S^{(\beta)}$ for each $\beta\geq s_L^+s_L^-$. Moreover: (i)
$\forall\,n\geq1$ $T_S^n(X_L)=Q_L\,T_S^n(X_L)\,Q_L$; (ii)
$\forall\,n\geq1$ $T_S^n(X_L)\in\A_S^{(\beta)}$; (iii)
$\forall\,n\geq1$ $\|T_S^n(X_L)\|\leq \beta^n\|X_L\|$.
\end{lemma}
\begin{proof}
We begin proving that $X_L\in\A_S^{(\beta)}$. For that it is
necessary to prove first that $X_L\in B(\Hil)\cap\LD$, which is
trivial since $\LD$ is an algebra and $X_L$ acts on a
finite-dimensional Hilbert space. Furthermore we have
$$
\|T_S(X_L)\|=\|S^{-1}Q_LXQ_LS\|=\|S^{-1}Q_LQ_LXQ_LQ_LS\|\leq$$
$$\leq\|S^{-1}Q_L\|\|Q_LXQ_L\| \|Q_LS\|\leq s_L^+s_L^-\|X_L\|\leq \beta
\|X_L\|
$$
so that $X_L\in\A_S^{(\beta)}$.

The proof of (i) follows from the fact that $Q_L$ is idempotent, $Q_L=Q_L^2$, and  that $Q_LS=SQ_L$ and $Q_LS^{-1}=S^{-1}Q_L$.

As for (ii) we just prove here that $T_S(X_L)\in \A_S^{(\beta)}$.
For higher powers the proof goes via induction. Since as we have
shown above $\|T_S(X_L)\|\leq \beta \|X_L\|<\infty$, $T_S(X_L)$ is a
bounded operator. Moreover, it clearly belongs to the algebra $\LD$.
Therefore we just need to check that
$\|T_S\left(T_S(X_L)\right)\|\leq \beta \|T_S(X_L)\|$. For this we have
$$
\|T_S\left(T_S(X_L)\right)\|=\|S^{-1}\left(S^{-1}Q_LXQ_LS\right)QS\|=
\|S^{-1}Q_L\left(S^{-1}Q_LXQ_LS\right)Q_LS\|\leq$$
$$\|S^{-1}Q_L\|\|S^{-1}(Q_LXQ_L)S\| \|Q_LS\|\leq
 s_L^+s_L^-\|T_S(X_L)\|\leq \beta \|T_S(X_L)\|
$$
This same estimate can be used to prove also (iii) above for $n=2$.
For that it is sufficient to recall also that $\|T_S(X_L)\|\leq \beta
\|X_L\|$. For $n\geq 3$
the proof goes again via induction.
\end{proof}

\vspace{2mm}

{\bf Remark:} One of the outputs of this Lemma is that, if we look
for elements of $\A_S^{(\beta)}$ for some fixed $\beta$, it is
enough to fix some value $L_0$ such that $s_{L_0}^+s_{L_0}^-$ is
less or equal to $\beta$. Then $Q_{L_0}XQ_{L_0}$ belongs to
$\A_S^{(\beta)}$ for each possible $X\in\LD$. This means, of course,
that the set $\A_S^{(\beta)}$ is rather rich, since each element of
$\LD$ produces, via suitable projection, an element of this set. Of
course, since both $s_{L}^+$ and $s_{L}^-$ increase with $L$, also
$Q_{L_0-1}XQ_{L_0-1}$ still belongs to $\A_S^{(\beta)}$, and so on.

\vspace{2mm}

We continue our list of results on the chain $\{\A_S^{(\alpha)},
\,\alpha>0\}$ giving the following density result:

\begin{cor}
Let $X\in\LD$. For all $\epsilon>0$ and $\forall\,(f,k)\in\C_0$
there exist $\beta>0$ and $\hat X\in\A_S^{(\beta)}$ such that
$\|X-\hat X\|^{f,k}<\epsilon$.
\end{cor}
\begin{proof}
We begin recalling that
$\|X-Q_LXQ_L\|^{f,k}\rightarrow0$ when $L\rightarrow\infty$, at
least if $\{s_l^{-n}\}$ belongs to  $l^2(\Bbb{N})$ for some $n> 0$. This means
that $\forall\,\epsilon>0$ and $\forall \,(f,k)\in\C_0$, there
exists $L_0>0$ such that, $\forall\,L\geq L_0$,
$\|X-Q_LXQ_L\|^{f,k}<\epsilon$. In particular, therefore, we have
$\|X-Q_{L_0}XQ_{L_0}\|^{f,k}<\epsilon$. But, for what we have shown
in the previous Lemma, $X_{L_0}:=Q_{L_0}XQ_{L_0}$ belongs to all
$\A_S^{(\beta)}$ with $\beta\geq s_{L_0}^+s_{L_0}^-$, so that if we
take $\hat X\equiv X_{L_0}$ our claim is proved.
\end{proof}

\vspace{2mm}

{\bf Remarks:} (1) This proposition suggests an approximation
procedure since any element of $\LD$ can be approximated, as much as
we like, with an element in some suitable $\A_S^{(\beta)}$.

(2) Of course, if instead of choosing $\hat X\equiv X_{L_0}$ we take
$\hat X\equiv X_{L_0+1}$ or yet $\hat X\equiv X_{L_0+2}$ (and so
on), we still obtain $\|X-\hat X\|^{f,k}<\epsilon$. However, in general, $X_{L_0+1}$  belongs to some $\A_S^{(\gamma)}$ with $\gamma>\beta$ and not necessarily to $\A_S^{(\beta)}$ itself.

\vspace{2mm}

What is an open problem at the present time is whether, at least for
some $\alpha$, $\A_S^{(\alpha)}$ coincides with all of
$B(\Hil)\cap\LD$ or not. In fact, there exist several results in
the literature which do not allow to conclude if this is true or
not, see \cite{kad,lass,bt2}. In our opinion no final argument exists
confirming or negating this fact. This  is at the basis of
the following Lemma, where the equality of the sets
$\A_S^{(\alpha)}$ and $B(\Hil)\cap\LD$ is simply assumed.

\begin{lemma}\label{lemma2}
Let us suppose that there exists $\alpha>0$ such that
$\A_S^{(\alpha)}=B(\Hil)\cap\LD$. Then $T_S:
\A_S^{(\alpha)}\rightarrow \A_S^{(\alpha)}$.
\end{lemma}
\begin{proof}
Let $X\in\A_S^{(\alpha)}=B(\Hil)\cap\LD$. Then, because of the
definition of $\A_S^{(\alpha)}$, $T_S(X)$ also belongs to
$B(\Hil)\cap\LD$ and therefore to $\A_S^{(\alpha)}$.
\end{proof}

\vspace{2mm}

{\bf Remark:} It would be interesting to prove, if possible, the
converse implication: is it true that if $T_S:
\A_S^{(\alpha)}\rightarrow \A_S^{(\alpha)}$ for some $\alpha$ then
$\A_S^{(\alpha)}=B(\Hil)\cap\LD$? Of course, it would be enough
to prove that for such an $\alpha$ the inclusion
$B(\Hil)\cap\LD\subseteq\A_S^{(\alpha)}$ holds true, since the
opposite inclusion is evident.

\vspace{2mm}

Lemmas \ref{lemma1} and \ref{lemma2} show that, under certain
conditions,  $T_S$ maps some $\A_S^{(\alpha)}$ (or some subset of it) into itself. The
following example shows that, more generally, there exist different
conditions under which $T_S$ maps some $\A_S^{(\alpha)}$ into a $\A_S^{(\beta)}$, in
general different. All these results motivate Proposition \ref{prop7} below.

\vspace{2mm}

{\bf Example 2:} Let us consider those elements $X\in
B(\Hil)\cap\LD$ satisfying, as in Example 1, the following
commutation relation: $[X,S]=BX$, with $B\in B(\Hil)\cap\LD$. As we
have seen, this implies that $X\in \A_S^{(\alpha)}$ with
$\alpha=\|\1+S^{-1}B\|$, so that $S^{-1}BS\in B(\Hil)\cap\LD$. We
also assume here that $S^{-2}BS^2\in B(\Hil)\cap\LD$. Under this
additional condition it is now easily checked that
$T_S(X)\in\A_S^{(\beta)}$, with $\beta=\|\1+S^{-2}BS^2\|$. Indeed we
can deduce that $T_S(X)\in B(\Hil)\cap\LD$ and, moreover, that
$\|T_S(T_S(X))\|\leq \beta \|T_S(X)\|$. Of course, it also immediately follows that $\|T_S(T_S(X))\|\leq \alpha\,\beta \|X\|$.

\vspace{2mm}

Because of the above example and using our previous considerations we now  prove the following result:
\begin{lemma}
Suppose that for all $\alpha>0$ there exists $\beta>0$ such that
$T_S:\A_S^{(\alpha)}\rightarrow\A_S^{(\beta)}$. Then,
$\forall\,X\in\A_S^{(\alpha)}$ and $\forall\,n\geq 1$, there exists
$\theta>0$ such that $\|T_S^n(X)\|\leq \theta \|X\|$.
\end{lemma}
\begin{proof}
We use induction on $n$ to  prove our claim.

For $n=1$ it is enough to choose $\theta=\alpha$, since
$X\in\A_S^{(\alpha)}$.

Let us now suppose that our claim is true for a given $n$. We need
to prove that the same statement holds for $n+1$. Indeed, since by
assumption $T_S(X)\in\A_S^{(\beta)}$ for $X\in\A_S^{(\alpha)}$, we
have $\|T_S^{n+1}(X))\|=\|T_S^n(T_S(X))\|\leq \theta\|T_S(X)\|\leq
\theta\alpha\|X\|=:\tilde\theta\|X\|$.
\end{proof}

The conclusion of this Lemma can also be restated by saying that, under
the same conditions stated above, the following inequality is
satisfied: \be \|S^{-n}XS^n\|\leq\theta\|X\|,\label{44}\en for each
$X\in\A_S^{(\alpha)}$. We refer to \cite{lass} and \cite{bt1} for a concrete realization of this inequality in the context of spin systems.

Using inequality (\ref{44}) we can prove the following result

\begin{prop}\label{prop7}
Suppose that for all $\alpha>0$ there exists $\beta>0$ such that
$T_S:\A_S^{(\alpha)}\rightarrow\A_S^{(\beta)}$. Let us further
assume that, given a sequence $\{X_j\}\subset\LD$, there exists
$n>0$ such that: (i) $S^{-n}X_j\in\A_S^{(\alpha)}$ for all $j$ and
some $\alpha>0$; (ii) $\|S^{-n}X_j\|\rightarrow0$ when
$j\rightarrow\infty$.

Then $\|S^{-n-k}X_jS^k\|\rightarrow0$ when $j\rightarrow\infty$ for
all $k\geq0$.
\end{prop}
We leave the easy proof to the reader. Here we apply this result to
the analysis of the algebraic dynamics as discussed before. In
particular we are now in a position of giving conditions which imply the main hypothesis of Proposition
\ref{prop22}.

\begin{cor}\label{cor1}
Suppose that for all $\alpha>0$ there exists $\beta>0$ such that
$T_S:\A_S^{(\alpha)}\rightarrow\A_S^{(\beta)}$. Let us further
assume that there exists $n>0$ such that: (i)
$S^{-n}\left(e^{iH_Lt}-e^{iH_Mt}\right)\in\A_S^{(\alpha)}$ for all
$L, M$ and some $\alpha>0$; (ii)
$\|S^{-n}\left(e^{iH_Lt}-e^{iH_Mt}\right)\|\rightarrow0$ when
$L,M\rightarrow\infty$.

Then $\|S^{-n-k}\left(e^{iH_Lt}-e^{iH_Mt}\right)S^k\|\rightarrow0$
when $L,M\rightarrow\infty$ for all $k\geq0$.
\end{cor}
This means that $S^{-n}$ has a {\em regularizing} effect: it is not
important whether the sequence $\left\{e^{iH_Lt}\right\}$ is Cauchy
in the uniform topology or not. What is important is that for some
$n>0$, and therefore for all $m\geq n$, the sequence
$\left\{S^{-m}e^{iH_Lt}\right\}$ is $\|.\|$-Cauchy.

\vspace{2mm}

We end this section by looking at conditions which implies that
$I_{L,M}^{(n)}:=\|S^{-n}\left(e^{iH_Lt}-e^{iH_Mt}\right)\|\rightarrow0$
when $L,M\rightarrow\infty$. For that we notice that
$g_{L,M}(t):=e^{iH_Lt}-e^{iH_Mt}$ can be written as \be
g_{L,M}(t)=i\int_0^t\,e^{iH_L(t-t_1)}\,(H_L-H_M)\,e^{iH_Mt_1}\,dt_1
\label{45}\en Therefore $I_{L,M}^{(n)}\leq
\int_0^t\left\|S^{-n}e^{iH_L(t-t_1)}\,(H_L-H_M)\right\|\,dt_1$, so
that the following Proposition is straightforwardly proved:
\begin{prop}\label{prop23}
Let $n$ be defined as in Corollary \ref{cor1}. Then
$I_{L,M}^{(n)}\rightarrow0$ for $L,M\rightarrow\infty$ in each of
the following conditions:

(1) \be \|S^{-n}(H_L-H_M)\|\rightarrow0\label{46}\en for
$L,M\rightarrow\infty$ and $[H_L,S^{-1}]=0$ for all $L$.

(2) Condition (\ref{46}) holds, there exists $\alpha>0$ such that
$e^{iH_Lt}\in\A_S^{(\alpha)}$ for all $L$ and
$T_S:\A_S^{(\alpha)}\rightarrow\A_S^{(\beta)}$, for some $\beta>0$.

(3) Condition (\ref{46}) holds and
$$
\sum_{k=1}^\infty\,\frac{\tau^k}{k!}\,\|S^{-n}[H_L,H_M]_k\|\rightarrow0
$$
for $L,M\rightarrow\infty$ and $\tau\geq0$.
\end{prop}
The proof of these statements is straightforward and is left to the
reader.

Summarizing, we have proved that condition (\ref{41}) is a
sufficient condition for the algebraic dynamics $\alpha^t$ to exist
in all of $\LD$. Moreover, sufficient conditions for (\ref{41}) to
be satisfied are discussed in Corollary \ref{cor1}. These conditions
are finally complemented by the results discussed in Proposition
\ref{prop23}. Also in this more general setting the same remark
following Proposition 1 can be restated: it may be convenient, in
some applications involving localized operators (e.g. for spin systems), to consider
a different approach and looking directly for the existence of the limit of
$\alpha_L^t(X)$, $X\in\LD$, instead of the existence of the limit of
$e^{iH_Lt}$.

Of course, due to the generality of the physical system we wish to
analyze, the conditions for the existence of $\alpha^t$ looks a bit
complicated. We refer to a paper in preparation for physical
applications of our strategy. Here we just consider few
applications, which are discussed in the next section.

\section{Physical applications}

In this section we discuss some examples of how our previous results
can be applied to physical systems. In particular, the first two
examples are related to what we have done in Section II, while the
last part of this section concerns the construction in Section III. Before starting however, we should say that the applications discussed in this section should be really considered as prototypes of real physical applications, since they are constructed using physical building blocks (bosonic operators) but no explicit expression for $H_L$ is given at all. This makes the following examples more general, from one point of view, but also not immediately related to  concrete physical systems. We will come back on this point in Section V.

\vspace{2mm}

{\bf Example 1:} We begin with a simple example, related to the {\em
canonical commutation relations}. Let $a$, $a^\dagger$ and
$N=a^\dagger\,a$ be the standard annihilation, creation and number
operators, satisfying $[a,a^\dagger]=aa^\dagger-a^\dagger\,a=\1$. Let $\varphi_0$ be the
vacuum of $a$: $a\varphi_0=0$, and
$\varphi_n=\frac{{a^\dagger}^n}{\sqrt{n!}}\,\varphi_0$. Then,
introducing the projector operators $P_l$ via the
$P_lf=<\varphi_l,f>\,\varphi_l$, we consider the following
invertible operator ($N$ is not invertible!):
$S=\1+N=\sum_{l=0}^\infty(1+l)P_l$. So, within this example, is the
number operator for a bosonic system which defines the algebraic and
the topological structure. We consider the following {\em
regularized} hamiltonian: $H_L=\sum_{l=0}^L\,h_l\,P_l$, where we
assume that $h_l$ is very rapidly increasing to $+\,\infty$. It is
obvious that we are in the setting of Section II, and in particular
that $[S,H_L]=0$ for all $L<\infty$. It can be deduced easily that,
$$
\left\|S^{-1}\left(e^{iH_Lt}-e^{iH_Mt}\right)\right\|\leq
\sqrt{\sum_{k=M+1}^L\,(1+k)^{-2}}\rightarrow0,
$$
assuming that $L>M$ and sending $L,M$ to infinity. In this case,
therefore, it is enough to choose $n=1$ in Proposition 1 and we see
explicitly that the sequence $\{h_l\}$ play no role in the analysis
of the dynamics: $\alpha^t$ can be defined for each choice of $\{h_h\}$.

\vspace{3mm}

{\bf Example 2:} Let us now discuss another example, related to an
infinitely extended physical system. We consider a lattice $\Lambda$
labeled by positive integers. To each lattice site $p$ is associated
a two-dimensional Hilbert space generated by two ortonormal vectors
$\varphi_p^{(0)}$ and $\varphi_p^{(1)}$. These vectors can be
constructed as follows: let $a_p$, $a_p^\dagger$ and
$N_p=a_p^\dagger\,a_p$ be the  annihilation, creation and number
operators, satisfying $\{a_p,a_p^\dagger\}=a_p,a_p^\dagger+a_p^\dagger a_p=\1_p$.
Now $\varphi_p^{(0)}$ is the vacuum of $a_p$:
$a_p\varphi_p^{(0)}=0$, while
$\varphi_p^{(1)}=a_p^\dagger\,\varphi_0$. Given a sequence
$\{n\}=\{n_1,n_2,n_3,\ldots\}$, $n_p=0,1$ for all $p\in\Bbb{N}$, we
can construct the following infinite tensor product vector:
$\varphi_{\{n\}}=\otimes_{p\in\Bbb{N}}\varphi_p^{(n_p)}$, which
belongs to the space $\Hil_\infty:=\otimes_{n\in\Bbb{N}}\Hil_n$. We
refer to \cite{thibook} for the details of the construction of an
infinite tensor product Hilbert space.

Let us introduce the set $\F:=\{\varphi_{\{n\}}\in\Hil_\infty \mbox{
such that } \sum_pn_p<\infty\}$, and let $\Hil$ be the Hilbert space
generated by all these vectors. Given a bounded operator on
$\Hil_p$, $X_p\in B(\Hil_p)$, we associate a bounded operator $\hat
X_p$ on $\Hil$ as follows:
$$
\hat X_p\varphi_{\{n\}}=\left(\otimes_{q\neq
p}\1_{q}\varphi_q^{(n_q)}\right)\otimes
\left(X_{p}\varphi_p^{(n_p)}\right)
$$
In this way we associate to $N_p$ an operator $\hat N_p$, and then
we construct $\hat N=\sum_p\hat N_p$. Of course, each vector in $\F$
is an eigenstate of $\hat N$: $\hat
N\varphi_{\{n\}}=n\varphi_{\{n\}}$, where $n=\sum_pn_p$. But for the
lowest eigenvalue, $n=0$, all the other eigenvalues are degenerate. For
instance, the eigenvectors associated to the eigenvalue $n=1$ are,
among the others, $\varphi_{\{n\}}^{(1,a)}:=\varphi_1^{(1)}\otimes
\varphi_2^{(0)}\otimes\varphi_3^{(0)}\otimes\cdots$,
$\varphi_{\{n\}}^{(1,b)}:=\varphi_1^{(0)}\otimes
\varphi_2^{(1)}\otimes\varphi_3^{(0)}\otimes\cdots$,
$\varphi_{\{n\}}^{(1,c)}:=\varphi_1^{(0)}\otimes
\varphi_2^{(0)}\otimes\varphi_3^{(1)}\otimes\cdots$ and so on. These
vectors are mutually orthogonal and define orthogonal projectors in
the standard way: $\varphi_{\{n\}}^{(1,a)}\rightarrow \hat
P_{1,a}f=<\varphi_{\{n\}}^{(1,a)},f>\,\varphi_{\{n\}}^{(1,a)}$,
$f\in\Hil$, and so on. Let now introduce $\hat P_n=\sum_{k\in
I_n}\hat P_{n,k}$, where $I_n$ is the set labeling the degeneration
of the eigenvalue $n$. The operator $\hat N$ can be written in terms
of these projectors as $\hat N=\sum_{n=0}^\infty\,n\,\hat P_n$.
Since $\hat N^{-1}$ does not exist, it is necessary to shift $\hat
N$ by adding, for instance, the identity operator. Therefore we put,
as in the previous example, $\hat S=\hat 1+\hat
N=\sum_{n=0}^\infty\,(1+n)\,\hat P_n$. Once again, if the
regularized hamiltonian $H_L$ can be written as
$H_L=\sum_{l=0}^L\,h_l\hat P_l$, no matter how the $h_l$'s look
like,
$\left\|S^{-1}\left(e^{iH_Lt}-e^{iH_Mt}\right)\right\|\rightarrow0$
when $L,M\rightarrow\infty$.

\vspace{3mm}

{\bf Example 3:} This example differs from the previous ones in that the projectors appearing in
the spectral decomposition of $S$ and $H_L$ are different. More explicitly we assume that $S=\sum_{l=0}^\infty s_lP_l$ and $H_L=\sum_{l=0}^L h_l\Pi_l$, with $P_l\neq \Pi_j$ for all $l$ and $j$. Of course, this implies that $[H_L,H_M]=0$ for all $L, M$, while, in general, $[S,H_L]\neq 0$. For simplicity reasons we consider here only the situation in which the relevant quantity to estimate is $S^{-n}\left(e^{iH_Lt}-e^{iH_Mt}\right)$ rather than $S^{-n+k}\left(e^{iH_Lt}-e^{iH_Mt}\right)S^{-k}$. Quite easy estimates allows us to write, for fixed $n$ and $L>M$,
$$
I_{L,M}:=\left\|S^{-n}\left(e^{iH_Lt}-e^{iH_Mt}\right)\right\|\leq
\sum_{l=0}^\infty\sum_{k=M+1}^L s_l^{-n}\,\|P_l\Pi_k\|
$$
It is clear that if $\Pi_k=P_k$ we would have $\|P_l\Pi_k\|=\delta_{l,k}$ and, if $\{s_l^{-n}\}\in l^1(\Bbb{N})$, then $I_{L,M}\rightarrow0$ when $L,M\rightarrow\infty$. This same conclusion can be deduced even under more general assumptions on $\Pi_k$. For instance, if we have $\|P_l\Pi_k\|=\sum_{j=0}^R\,\beta_j^{(l,k)}\delta_{l,k+j}$ for some finite $R$ and for a set of real coefficients $\beta_j^{(l,k)}$, we can  estimate $I_{L,M}$ as follows:
$$
I_{L,M}\leq \sum_{k=M+1}^L\left(s_k^{-n}+s_{k+1}^{-n}+\cdots+s_{k+R}^{-n}\right)
$$
and the right-hand side goes to zero when $L,M\rightarrow\infty$. This condition reproduces explicitly what has been already discussed in Section II.1.

More interesting is the situation in which $R=\infty$ in the previous formula, i.e. when $\|P_l\Pi_k\|=\sum_{j=0}^\infty\,\beta_j^{(l,k)}\delta_{l,k+j}$ for some  set of real coefficients $\beta_j^{(l,k)}$. First we remark that the right-hand side is not an infinite sum: on the contrary, for fixed $k$ and $l$ it collapses in a single contribution. Moreover, since in any case $\|P_l\Pi_k\|\leq \|P_l\|\,\|\Pi_k\|=1$, we deduce that: $0\leq\beta_j^{(k+j,k)}\leq 1$ for all $k,j\geq0$. Under this assumption we get $I_{L,M}\leq \sum_{j=0}^\infty\,\sum_{k=M+1}^L\,s_{k+j}^{-n}\,\beta_j^{(k+j,\,k)}$. It is now easy to find conditions on $\beta_j^{(k+j,\,k)}$ which imply that $I_{L,M}\rightarrow0$ when $L,M\rightarrow\infty$. The first possible condition is the following:
\be
\beta_j^{(l,k)}=\beta_j^{(l-k)}\quad\mbox{and}\quad \beta_j:=\beta_j^{(j)}=\beta_j^{(k+j,\,k)}\in l^1(\Bbb{N})
\label{ex1}\en
In this case, since the sequence $\{s_k^{-n}\}$ belongs to $l^1(\Bbb{N})$, from some $k$ on we have $s_{k+j}^{-n}\leq s_k^{-n}$ for all $j\geq0$. Hence we get
$$
I_{L,M}\leq \sum_{j=0}^\infty\,\sum_{k=M+1}^L\,s_{k}^{-n}\,\beta_j\leq \|\beta\|_1\,\sum_{k=M+1}^L\,s_{k}^{-n}
$$
which goes to zero for $L$ and $M$ diverging.

A different possibility which again implies the same conclusion is the following: suppose that for each fixed $k$ the series $\sum_{j=0}^\infty\left(\beta_j^{(k+j,\,k)}\right)^2$ converges to some $B_k$ such that $\{B_k\}$ belongs to $l^1(\Bbb{N})$. Then again $I_{L,M}\rightarrow0$ when $L,M\rightarrow\infty$. Indeed we have, since $s_{k+j}^{-n}\leq s_j^{-n}$ for all $k\geq0$ and for each fixed $j$, $I_{L,M}\leq \sum_{j=0}^\infty\,s_j^{-n}\,b_j^{(L,M)}$, where $b_j^{(L,M)}:=\sum_{k=M+1}^L\,\beta_j^{(k+j,\,k)}$. Then, using Schwarz inequality, we get $$I_{L,M}\leq \sqrt{\sum_{j=0}^\infty\,s_j^{-2n}}\,\sqrt{\sum_{j=0}^\infty\,\left(b_j^{(L,M)}\right)^2},$$
which shows that $I_{L,M}$ goes to zero if and only if $J_{L,M}:=\sum_{j=0}^\infty\,\left(b_j^{(L,M)}\right)^2$ goes to zero. This is because $\sum_{j=0}^\infty\,s_j^{-2n}$ surely converges, in our hypothesis. Using Schwarz inequality again it is finally easy to check that $J_{L,M}\leq \sum_{k=M+1}^L\,B_k$, which goes to zero because of our assumption on $B_k$.

It is worth stressing that the two different assumptions considered so far are mutually excluding. In particular, it is clear that if (\ref{ex1}) holds, then $\{B_k\}$ cannot belong to $l^1(\Bbb{N})$.

\section{Conclusions}

In this paper we have generalized some of our previous results on
the existence of the algebraic dynamics within the framework of
$O^*$-algebras. In particular we have considered two different
{\em degrees of generalizations}: in the first step the topology, the
$O^*$-algebra and the regularized hamiltonians  all arise from
the same spectral family $\{P_l\}$. We have shown that for each
$X\in\LD$ its time evolution $\alpha^t(X)$ can be defined under very
general assumptions. Secondly we have considered a really different
situation, in which there is no a priori relation between the
topological quasi *-algebra $\LD[\tau]$ and the family of
regularized hamiltonians. In this case the definition of $\alpha^t$
is much harder but can still be achieved, at least on a large subset
of $\LD$.

Some preliminary physical applications have been considered here, while others are
planned in a future paper, where we also hope to consider in detail
the problem of the existence of other global quantities, like the
KMS states, the entropy of the system, and others. We also will slightly modify our point of view, considering directly the thermodynamical limit of $\alpha_L^t(X)=e^{iH_Lt}Xe^{-iH_Lt}$ rather than the limit of $e^{iH_Lt}$. This, we believe, can enlarge the range of applicability of our results since the power expansion of  $\alpha_L^t(X)=X+it[H_L,X]+\cdots$ involves some commutators which, in general, behave much better than $H_L$ itself when the limit for $L\rightarrow\infty$ is taken. In particular, it should be mentioned that (almost) mean-field spin models should be treated with this different technique, see \cite{bt1,bt2}, rather than with the approach proposed in this paper, which is more convenient for those systems discussed in \cite{bag1}.

\bigskip
\noindent {\large\bf Acknowledgement}

This work has been financially supported in part by M.U.R.S.T.,
within the  project {\em Problemi Matematici Non Lineari di
Propagazione e Stabilit\`a nei Modelli del Continuo}, coordinated by
Prof. T. Ruggeri. I also like to thank Prof. Timmermann because a
discussion we had few years ago during a conference in Bedlewo
pushed me to work on this subject. I also thank him for his comments and suggestions. Finally, I also thank Prof. C. Trapani for pointing my attention to \cite{kad}.

%\vfill

%\hyperlink{label}{  prova}
%\hypertarget{vonNeumann-1}{[back]}

%%%%%%%%%%%%%%%%%%%%%%%%% APPENDICI MATEMATICHE ED INFORMAZIONI CHE POSSONO ESSERE UTILI %%%%%%%%%%%%%%%%%%%%

\end{document}